\documentclass[11pt]{article}
\usepackage{amsmath,amsthm,amsfonts,amssymb,enumerate,calc,graphicx,color}
\usepackage{libertine, newtxmath}
\usepackage[colorlinks=true,citecolor=black,linkcolor=black,urlcolor=blue]{hyperref}
\usepackage[lmargin=28mm,rmargin=28mm,bmargin=30mm,tmargin=30mm]{geometry}

\usepackage{paralist}

\usepackage{graphicx,amsmath}
\usepackage{url}

\newtheorem{thm}{Theorem}{\bfseries}{\itshape}
\newcommand{\thmlabel}[1]{\label{thm:#1}}
\newcommand{\thmref}[1]{Theorem~\ref{thm:#1}}

{\bfseries}{\itshape}
\newcommand{\lemlabel}[1]{\label{lem:#1}}
\newcommand{\lemref}[1]{Lemma~\ref{lem:#1}}

\title{A note on choosability with defect 1 of graphs on
  surfaces}

\author{Vida Dujmovi\'c$^*$ \and Djedjiga Outioua\thanks{School of
    Electrical Engineering and Computer Science,
        University of Ottawa, emails: vida.dujmovic@uottawa.ca, douti102@uottawa.ca. Research supported by NSERC and
      Ministry of Research, Innovation and Science of Ontario. 
       The first author would like the thank
                    Fr\'ed\'eric Havet for introducing her to
                    defective list colourings more than a decade ago,
               during her visit to INRIA Sophia Antipolis,
                 funded by French Consulat G\'en\'eral de France
                    \'a Montr\'eal.
}}

\theoremstyle{plain}

\newtheorem{lemma}{Lemma}

\theoremstyle{definition}

\newcommand{\ceil}[1]{\ensuremath{\protect\lceil#1\rceil}}
\newcommand{\bceil}[1]{\ensuremath{\protect\Big\lceil#1\ \Big\rceil}}

\newcommand{\floor}[1]{\ensuremath{\protect\lfloor#1\rfloor}}

\newcommand{\EndProof}[1]{
\begin{minipage}[b]{\textwidth}
#1
\end{minipage}\hfill\qed
\renewcommand{\qed}{}
}

\newcommand{\eg}[1]{\ensuremath{\protect\textup{\textsf{eg}}(#1)}}

\newcommand{\ch}[1]{\ensuremath{\protect{(#1)^*}-}}
\newcommand{\cl}[1]{\ensuremath{\protect{(#1)^*-}}}

\renewcommand{\deg}[2][]{\ensuremath{\protect{\textup{\textsf{deg}}_{#1}(#2)}}}

\newcommand{\w}[1]{\ensuremath{\protect{\textup{\textsf{w}}(#1)}}}

\newcommand{\Ne}{\ensuremath{\protect{\mathbb{N}}}}

\newcommand{\sm}{\ensuremath{\setminus}}
\newcommand{\se}{\ensuremath{\subseteq}}

\date{~}
\begin{document}
\maketitle

\begin{abstract}
This note proves that every graph of Euler genus $\mu$ is $\ceil{2 +  \sqrt{3\mu + 3}}$--choosable with defect 1 (that is, clustering
2). Thus, allowing defect as small as 1 reduces the choice number of
surface embeddable graphs below the chromatic number of the surface. For example,
the chromatic number of the family of toroidal graphs is known to be
$7$. The bound above implies that toroidal graphs are $5$-choosable
with defect 1. This strengthens the result of  Cowen, Goddard and
Jesurum (1997)  who showed that toroidal graphs are $5$-colourable
with defect 1.
\end{abstract}

\section{Introduction}

In a vertex coloured graph, a \emph{monochromatic component} is a
connected component of $G$ where each vertex has the same
colour. A vertex colouring is \emph{proper} if every monochromatic
component has $1$ vertex. The following notion, introduced formally by
Cowen, Cowen and Woodall \cite{DBLP:journals/jgt/CowenCW86} and
studied as early as 1966 by Lov\'asz~\cite{l66}, generalizes proper graph colourings. A
graph $G$ is $k$-colourable with \emph{defect} $d$, that is, \emph{\ch{k,d}colourable}, if the vertices of $G$ can be coloured with $k$ colours such that the subgraph
induced by the vertices of each monochromatic component has maximum
degree $d$. Colourings with defect zero are proper, otherwise they are
called  \emph{defective}. In this note we are
interested in defective colourings that
are as close as possible to being proper, that is, 
in colourings with defect at most 1. In particular, we study such
colourings for graphs on surfaces.  Typically, the goal of this line of
research is to reduce the number of colours below the number required
by proper colourings.  For example, Voigt proved~\cite{DBLP:journals/dm/Voigt93} that there are planar graphs that are
not 4-choosable, but Cushing and
Kierstead \cite{DBLP:journals/ejc/CushingK10} proved that planar
graphs are  4-choosable if defect 1 is allowed,  thereby answering an 
open problem posed by several authors \cite{DBLP:journals/jgt/Woodall04,aop, Skre,Eaton}. Defective
choosability will be formally defined below. 
For some classes of graphs, however, allowing even arbitrarily big
defect does not reduce the number of colours below the chromatic
number of the class. For example, $k$-trees have chromatic number $k+1$,
and yet for every $d$ there is $k$-tree that is not \ch{k,d}colourable (see the standard
example in \cite{wood-survey-2018}). 
This implies that, for any $d$, there
is a planar graph (in fact a   series-parallel graph), that is not
$(2,d)$-colourable. 

A similar notion to that of defective colouring is clustered
colouring.  A graph $G$ is $k$-colourable with \emph {clustering} $c$ if the
vertices of $G$ can be coloured with $k$ colours such that each
monochromatic component has at most $c$ vertices. Note that a colouring of a graph has defect at most $1$ if and only if it has
clustering at most $2$. This is not the case for higher
defect/clustering values. Thus, our results also give colourings of
graphs on surfaces with clustering at most $2$. There is a plethora of
recent work on the subject
of clustered and defective graph colourings/choosability. For an extensive coverage of the topic see the recent survey by
Wood \cite{wood-survey-2018} and Sections 3, 4 and 5 in the survey by
Woodall \cite{Woodall-survey-2001}.


There is a rich history of various colouring problems on
graphs on surfaces. For the background, we only focus on the previous work on
colourings and choosability with defect 1.  The goal of this note is to
show that for every surface, graphs embeddable in that
surface have choice number with defect 1 less than the chromatic
number of that surface. 

A surface $\Sigma$ is a compact connected 2-manifold without
boundary. Surfaces are classified into two classes. Each orientable
(non-orientable) surface is homeomorphic to a sphere with $g\geq 0$
handles ($h\geq 1$ crosscaps) attached. For a surface $\Sigma$, its
\emph{Euler genus}, \eg{\Sigma}, is $2g$ if $\Sigma$ is orientable and
$h$ if $\Sigma$ is non-orientable. Given a graph $G$, its
Euler genus \eg{G} is the minimum Euler genus of a surface $\Sigma$
where $G$ can be embedded in. 

In their 1997 paper, Cowen, Goddard, and Jesurum \cite{DBLP:conf/soda/CowenGJ97,
  DBLP:journals/jgt/CowenGJ97} studied defective colouring of
graphs on surfaces. In their conclusions, they suggest a study of
defective choosability, that is, a list version of the problem. This is what we study in
this note, namely rather than studying the colouring problem of graphs
on surfaces while allowing defect $1$, we study a list version of the problem which is its strengthening. For
planar graphs, the defective list version was first studied by  Eaton and Hull~\cite{Eaton} and
\v{S}krekovski~\cite{Skre}.

A \emph{list assignment} for a graph $G$ is a function $L$ that
assigns a set $L(v)$, also called a \emph{list} $L(v)$, of colours to
each vertex $v \in V (G)$.  A list $L(v)$ with $|L(v)| \geq k$ is a
$k$-\emph{list}. A list assignment $L$ is a $k$-\emph{list assignment} if each
vertex is assigned a $k$-list. For a list assignment $L$, $G$ is $L$-\emph{colourable} with defect $d$ if
$G$ can be coloured such that each vertex $v$ gets a colour from its
list $L(v)$ and such that the defect is at most
$d$. $G$ is \ch{k,d}\emph{choosable} if for every $k$-list
assignment $L$,  $G$ is $L$-colourable with defect $d$.  In an
$L$-colouring of a graph, a vertex $v$ is \emph{proper} if none of its
neighbours have the same colour as $v$. The
\emph{choice-$i$ number} of a graph $G$ is the minimum $k$ such that $G$ is \ch{k,i}choosable.
The \emph{chromatic-$i$  number} of a graph $G$ is the minimum $k$
such that $G$ is \ch{k,i}colourable. For proper colourings, $i=0$ is
omitted in this notation. 

The following is our main result.

\begin{thm}\thmlabel{full}
Every graph $G$ is $\bigg(\bceil{2 + \sqrt{3\eg{G} + 3}}, 1\bigg)^*$--choosable.
\end{thm}

 Cowen, Goddard and Jesurum~\cite{DBLP:journals/jgt/CowenGJ97,DBLP:conf/soda/CowenGJ97} proved that toroidal graphs are \cl{5,1}colourable. \thmref{full} implies that they are, in fact, \ch{5,1}choosable thus the theorem constitutes a strengthening of that result. In addition, the proof of \thmref{full} does not use the four-colour theorem.
 For planar graphs, the theorem states that they are
\ch{4,1}choosable. This has been proved by  Cushing and
Kierstead \cite{DBLP:journals/ejc/CushingK10}. Our proof works for \ch{t,1}choosability where $t\geq 5$,
thus it does not imply Cushing and Kierstead's result.

It is well-known (by Heawood's conjecture  \cite{Genus1}  proved by Ringel and
Young \cite{Genus2}) that for every surface $\Sigma$, except the Klein
bottle, there is a graph (in fact a complete graph) that embeds in that surface whose
chromatic number is exactly $\floor{\frac{7+\sqrt{24\eg{\Sigma}+1}}{2}}=\floor{3.5+\sqrt{6\,\eg{\Sigma}+0.25}}
$. Since this lower bound is achieved for complete graphs $G$, it follows
that the choice number of such graphs is at least
$\floor{3.5+\sqrt{6\,\eg{G}+0.25}}$. Thus, \thmref{full} shows
that for every surface, the choice-$1$ number of graphs
embeddable in the surfaces can be reduced below what the chromatic number
of the surface allows. 

The bound in \thmref{full} is surely not tight. The only lower bound
available however, follows from the fact that having choice-1 number at most $p$ implies having a chromatic number
at most $2p$. Thus, the above $\floor{3.5+\sqrt{6\,\eg{\Sigma}+0.25}}$
bound on the chromatic number implies that for every surface $\Sigma$, except the Klein
bottle, there is a graph  that embeds in that surface whose
choice-1 number is at  least $\ceil{1.75 + \sqrt{1.5\,\eg{G} + 1/16}}$. For example, this lower bound and
\thmref{full} imply that the correct bound on choice-1 number of toroidal graphs is either 4 or 5, leading to an open
question of whether toroidal graphs are \ch{4,1}choosable.  
Unfortunately it is still not even known if toroidal graphs are
\ch{4,1}colourable, which is  an open
question from 1997 by Cowen, Goddard, and Jesurum \cite{DBLP:conf/soda/CowenGJ97,
  DBLP:journals/jgt/CowenGJ97}. 

Even the question on weather all planar graphs are \ch{4,1}choosable
has been open until recently.  The question, which received
considerably attention, was asked in several articles over
the years \cite{DBLP:journals/jgt/Woodall04,aop, Skre, Eaton} and was finally settled in the positive by Cushing and
Kierstead \cite{DBLP:journals/ejc/CushingK10}. A well-known result by Voigt states that there are planar graphs that are not 4-choosable
\cite{DBLP:journals/dm/Voigt93}. Cowen, Cowen and Woodall
\cite{DBLP:journals/jgt/CowenCW86} proved that there are planar graphs
that are not \ch{3,1}colourable and thus they are not
\ch{3,1}choosable.  Cowen, Goddard, and Jesurum
\cite{DBLP:journals/jgt/CowenGJ97} show that in fact testing if a planar
graph is \ch{2,1}colourable or \ch{3,1}colourable is
NP-complete. The hardness of \ch{2,1}colourability remains unchanged even
for planar graphs with maximum degree 4, as proved by  Corr{\^{e}}a, Havet and
Sereni~\cite{DBLP:journals/ajc/CorreaHS09}.  For more on the
complexity 
of defective colouring problems, see the recent results by Belmonte, Lampis and Mitsou~\cite{DBLP:conf/stacs/BelmonteLM18}.

Note that the above negative result on  \ch{3,1}colourability of
planar graphs implies that  the aforementioned lower bound,
$\ceil{1.75 + \sqrt{1.5\,\eg{G} + 1/16}}$, on chromatic-1 and choice-1
number is not tight, in fact it is off by $2$ for planar graphs. Thus, the lower bound is
possibly quite weak for choice-1 number for graphs on surfaces. We
conclude this section by asking for improvements on the lower and upper
bound on choice-1 number of such graphs, starting with the
question of whether toroidal graphs are \ch{4,1}choosable or 
\ch{4,1}colourable. Colouring and choosability problems often benefit
from proving a non-trivial maximum average degree results for a
colouring problem. Unfortunately, while such results are known for
defective-$d$ choosability with $d\geq2$, by the results of Havet and
Sereni \cite{DBLP:journals/jgt/HavetS06}, no such results are known
for defective-$1$ choosability. Any non-trivial bound on maximum average degree for defective-$1$
choosability, could likely be used to improve the result in \thmref{full}.

\section{Useful Lemmas}
We consider simple and undirected graphs $G=(V,E)$ with vertex set
$V(G)$ and edge set $E(G)$. The number of vertices is denoted by
$n=|V(G)|$ and the number of edges by $m=|E(G)|$. For a set $S\se
V(G)$, $G[S]$ denotes the subgraph of $G$ induced by the vertices of
$S$ and $G-S$ denotes the subgraph $G[V(G)\sm S]$. If $S$ is comprised
of one vertex, $v$, then $G-v$ denotes $G-\{v\}$. For every $v\in V$, let $N(v)$ denote the set of neighbours of $v$. The {\emph{degree} of $v$ is $\deg[G]{v}=|N(v)|$.  A vertex of degree $d$ is called \emph{degree--$d$} vertex.  
The minimum and maximum vertex degree in $G$ are denoted respectively
by $\delta(G)$, and $\Delta(G)$. We omit ``$G$'' from this notation whenever the graph is clear from the context.


We start with a useful observation. Lov\'asz'~\cite{l66} proof for
defective colouring of graphs of maximum degree $\Delta$ extends
easily to choosability. We include the proof for completeness.

\begin{lemma}\textup{~\cite{l66}}
\lemlabel{lovaz}
For every integer $k>0$, every graph $G$ with maximum degree $\Delta$
is \ch{k, \floor{\frac{\Delta}{k}}} choosable. In fact, $G$ is $k$-choosable such that each vertex $v\in V(G)$ has at most $\floor{\frac{\deg{v}}{k}}$ neighbours with the same colour as $v$.
\end{lemma}

\begin{proof}
For a $k$-list assignment $L$ of $G$  consider an $L$--colouring of
$G$ that minimizes the number of monochromatic edges. Assume for the
sake of contradiction that this is not a desired list colouring. Let
$V_1, V_2, \dots V_{|L(G)|}$ denote the resulting colour classes (some
possibly empty), where $L(G) = \cup\{L(v) | v \in V(G) \}$. Assume that that there is a vertex $v\in V(G)$ coloured $c\in L(v)$ such that there are at least $\floor{\frac{\deg{v}}{k}}+1$ neighbours of $v$ in the colour class $V_c$. In that case, there is a colour class $V_p$, $p\in L(v)$ such that the number of neighbours of $v$ in $L_p$ is at most $\floor{\frac{\deg{v}}{k}}$. Changing the colour of $v$ from $c$ to $p$ reduces the number of monochromatic edges, thus the contradiction.
\end{proof}

The proof of \thmref{full} uses the discharging technique. The following lemma provides some key observations that will be used for discharging rules.

\begin{lemma}\lemlabel{prop}
Let $G$ a vertex minimal graph 
of Euler genus at most $\mu$ 
 such that $G$ is not \ch{t,1}choosable, $3\geq t\in \Ne$. Then $G$ has the following properties.

\begin{compactenum}
\item
\begin{compactenum}
\item\label{min} The minimum degree, $\delta(G)$, of $G$ is at least $t$
\item\label{max} The maximum degree, $\Delta(G)$, of $G$ is at least
  $2t$
\item\label{numv} The number of vertices in $G$ is at least $2t+1$.
\end{compactenum}
\item
\begin{compactenum}
\item \label{IS} The set of degree-$t$ vertices of $G$ forms an independent set in $G$.
\item \label{degt1} Each vertex $v$ of $G$ has at most
  \floor{\frac{\deg{v}}{2}} degree-$t$\ neighbours, whenever $G$ is
  edge maximal.
\end{compactenum}

\item 
\begin{compactenum}
\item \label{triangle} There is no 3-cycle  $v,w,u$ in
  $G$ such that $\deg{v}=t$, $\deg{w}=t+1$ and $
\deg{u}=t+1$.

\item  \label{degt} Each degree-$t$ vertex $v$ of $G$, has at least \ceil{\frac{t}{2}} degree-$d$,\ $d\geq t+2$, neighbours whenever $G$ is edge-maximal.
\end{compactenum}

\item  \label{big} 
Whenever $G$ is edge-maximal, if $t$ is even and a degree--$t$ vertex $v$ has exactly $\frac{t}{2}$ degree-$d$,\  $d\geq t+2$, neighbours, then at least one of them has degree at least $t+3$.


\end{compactenum}

\end{lemma}

 \begin{proof}
By the assumptions of the lemma there is a  $t$-list assignment $L$ such
that $G$ is not $L$-colourable. 

\ref{min}. Assume on the contrary that $G$ has a vertex $v$ of degree at most
$t-1$.  Then $G-v$ is a nonempty graph of Euler genus at most $\mu$ and thus
by  the vertex minimality of $G$ it is  $L$-colourable with defect $1$. 
Since $|N(v)|\leq t-1$, there is a colour in the $t$-list $L(v)$ that is not used
by any vertex in $|N(v)|$, thus the $L$-colouring of $G-v$ can be
extended to $L$-clouring of $G$, giving the contradiction. 

\ref{max}. If $\Delta(G)<2t$, then $\floor{\frac{\Delta(G)}{t}}=1$ and thus
$G$ is \ch{t,1}choosable  by \lemref{lovaz}, thus contradicting the
assumptions of the lemma.

\ref{numv}. Follows from \ref{max}.

\ref{IS}.  Assume on the contrary that $G$ has two degree-$t$ vertices,
$v$ and $w$, that are adjacent. Then $G-\{v,w\}$ is a nonempty graph of Euler genus at most $\mu$ and thus
by the vertex minimality of $G$ it is  $L$-colourable  with defect $1$. Since
$|\{N(v)-w\}|\leq t-1$ and $|\{N(w)-v\}|\leq t-1$, there is a
colour $c_1$ in the $t$-list $L(v)$ and $c_2$  in the $t$-list $L(w)$
such that no vertex in $\{N(v)-w\}$ is coloured $c_1$ and no vertex in $\{N(w)-v\}$ is coloured $c_2$ in the $L$-colouring of $G-\{v,w\}$. Thus,
assigning colour $c_1$ to $v$ and $c_2$ to $w$ gives $L$-colouring of
$G$ with defect 1 (even if $c_1=c_2$).

\ref{degt1}. We now can assume $G$ is edge maximal graph with a 2-cell
embedding in a surface of Euler genus at most $\mu$. The cyclic
ordering of edges around $v$ in the embedding. By edge maximality of
$G$, for any pair of consecutive edges $vw$ and $vy$ around $v$, we have that $x$ and $y$ are adjacent. Thus, by~\ref{min} and \ref{IS},  $x$ or $y$ must have degree at least $t+1$, implying the
claim.

\ref{triangle}. $G-\{v,w,u\}$ is
a nonempty graph of Euler genus at most $\mu$ and thus by the vertex
 minimality, $G$ it is  $L$-colourable  with defect $1$.  Since $|\{N(w)-\{v,u\}\}|\leq t-1$ and $|\{N(u)-\{v,w\}\}|\leq t-1$, there is a
 colour $c_1$  in the $t$-list $L(w)$ and $c_2$  in the $t$-list
 $L(u)$ such that no vertex in $\{N(w)-\{v,u\}\}$ is coloured $c_1$ and no vertex in
$\{N(u)-\{v,w\}\}$ is coloured $c_2$ in the $L$-colouring of
$G-\{v,w,u\}$. Assign colour $c_1$ to $w$ and $c_2$ to $u$. That gives
$L$-colouring of $G-v$ with defect 1. If  there is a colour $c$ in the
$t$-list $L(v)$ such that no vertex in $N(v)$ is coloured $c$ in that
$L$-colouring of $G-v$, then assigning $v$ colour $c$ gives
$L$-colouring of $G$ with defect 1. Otherwise, $c_1\in L(v)$ and
$c_2\in L(v)$ and $c_1\not= c_2$.  
In that case, both $w$ and $u$ are proper in  the $L$-colouring of
$G-v$ and thus assigning $v$ colour $c_1$ gives $L$-colouring of
$G$ with defect 1.

\ref{degt}.  We now can assume $G$ is edge maximal graph with a 2-cell
embedding in a surface of Euler genus at most $\mu$. The edge maximality of $G$ implies
that every pair of consecutive edges $vx$ and $vw$ around $v$, defines
a 3-cycle $v,x,y$ in $G$. Then \ref{triangle} and the fact that
$\deg{v}=t$ imply the claim.

\ref{big}. We now can assume $G$ is edge maximal graph with a 2-cell
embedding in a surface of Euler genus at most $\mu$. Assume for the
sake of contradiction that all the neighbours of $v$ have degree at most $t+2$. Let
$H=G[N(v)]$ and $H'=G[v\cup N(v)]$. Consider vertices in $N(v)$ in the
cyclic order, $v_1, v_2\dots v_t$, around $v$ as determined by the
embedding and starting with a degree--$(t+2)$ vertex $v_1$.  By the
edge maximality of $G$, $C=v_1, \dots, v_t, v_1$ is a cycle (non
necessarily induced) in $G$. Then the degrees in $G$ of vertices in
$N(v)$ are as follows $\deg{v_i}=t+2$ for $i\equiv (1 \mod 2)$ and
$\deg{v_i}=t+1$ for $i\equiv (0 \mod 2)$. (The fact that the degrees
alternate, between $t+1$ and $t+2$, along $C$ is the consequence of Property \ref{IS},
Property \ref{triangle} and the assumption that $v$ has exactlly $\frac{t}{2}$
degree-$(t+2)$ neighbours). Thus, since $t$ is even, for each
degree-$(t+2)$ vertex in $C$ its two neighbours along $C$ are
degree-$(t+1)$ vertices. $G'=G-\{v\cup N(v)\}$ has genus at most
$\mu$ and it is not an empty graph by Property \ref{numv} and the fact
that $|v\cup N(v)|=t+1$, and $t>0$. Thus, $G'$ is $L$-clourable with
defect $1$. This list colouring of $G'$ can be extended to a
$L$-colouring of $G$ with defect $1$ as follows.

Define the list assignment $L'$ of $H'$ (and $H$) as follows. For
every $w\in  V(H')$ the list $L'(w)$ is equal to the list $L(w)$ minus the colours
used by the neighbours of $w$ in $L$-colouring of $G'$. Clearly any $L'$-list colouring
of $H'$ with defect $1$, extends the colouring of
$G'$ to $L$-colouring of $G$ with defect 1. By considering degrees in $G$
of vertices in $H'$ we get, $|L'(v)|=|L(v)|$, $|L'(v_i)|\geq
\deg[H']{v_i}-2$ for $i\equiv (1\mod 2)$ and  $|L'(v_i)|\geq
\deg[H']{v_i}-1$ for $i\equiv (0\mod 2)$. Moreover, in $H$,
$|L'(v_i)|\geq \deg[H]{v_i}-1$ for $i\equiv (1\mod 2)$ and
$|L'(v_i)|\geq \deg[H]{v_i}$ for $i\equiv (0\mod 2)$. Let $S=\{\cup
v_i\, |\, i\equiv (0\mod 2)\}$. By the above observation, each vertex $v_i$, $i\equiv (1\mod
2)$, is adjacent to at least 2 vertices of $S$ (its neighbours along $C$), thus
$\deg[H]{v_i}\geq \deg[H-S]{v_i}+2$.  
Thus, in $H- S$, $|L'(v_i)|\geq \deg[H-S]{v_i}+1$ for all $i\equiv (1\mod
2)$. Therefore, each vertex $x$ in $H-S$ has smaller degree in $H-S$
than the number of colours in its list as determined by $L'$ list
assignment, that is, $\deg[H-S]{x}<|L'(x)|$. The greedy colouring then 
implies that  $H- S$ has $L'$-colouring where every vertex in $H-S$ is
proper. 

By property \ref{triangle}, $S$ forms an independent set in
$H$. Thus, all of $\deg[H]{v_i}$ neighbours of vertex $v_i\in S$ in $H$ are in
$V(H)- S$. Let $A$ denote the vertex set comprised of the vertices
$v_i\in S$ whose neighbours in $H- S$ use at most $|L'(v_i)|-1\geq\deg[H]{v_i}-1$
colours. Let $B=S- A$. Since  $|L'(v_i)|\geq \deg[H]{v_i}$ for each vertex $v_i\in S$, each vertex in $v_i\in A$ can choose a colour from its list $L'(v_i)$ such that $v_i$ is proper in $L'$-colouring of
$H- B$ (and in $H$ as will be seen later). Then $H- B$ is
$L'$-coloured such that every vertex of $H- B$ is proper. For each vertex $v_i\in B$, each of its colours in $L'(v_i)$
is used by exactly one of its neighbours in $H$, and each of its
neighbours in $H$ uses actually one colour in $L'(v_i)$. Thus, giving
each such vertex $v_i$ the colour equal to the colour of its first
counterclockwise neighbour along $C$ defines $L'$-colouring of
$H$. It remains to show that this $L'$-colouring of $H$ has
defect  at most 1. Clearly, the vertices of $A$ are proper in the
colouring of $H$. Vertices of $H- S$ are proper in the colouring of
$H- B$, thus  all the monochromatic edges have one endpoint in $B$ and
the other in $V(H)- S$. Assume a vertex $w\in B$ has two neighbours in
$V(H)- S$ coloured with the same colour as $w$. That implies that the
number of colours used by the neighbours of $w\in V(H)- S$ is at most
$\deg[H]{w}-1$ thus $w\in A$, contradiction. Finally, assume a vertex
$w\in V(H)- S$ has two neighbours $x$ and $y$ in $B$ coloured with the
same colour as $w$. That implies that $w$ is the first counter clockwise
neighbour on $C$ of both $x$ and $y$, which is impossible since no
pair of vertices of $S$ are adjacent.  Thus, we have an $L'$-colouring
of $H$ with defect 1. Recall that $|L'(v)|=t$. In $H'$, if the
vertices in $H$ do not use all the colours in $L'(v)$, then $v$ can be
coloured such that it is proper in $L'$-colouring of $H'$, thus resulting in the $L'$--colouring of $H'$ with defect at most one. Otherwise, each colour in $L'(v)$ is used by exactly one vertex in $H$ and each vertex in  $H$ uses exactly one colour in $L'(v)$. Therefore, the $L'$--colouring of $H$ is proper and no two neighbours of $v$ use the same colour. Thus, $v$ can use any colour in  $L'(v)$ and the resulting $L'$-colouring of $H'$ has defect at most one.

By the choice of $L'$ the resulting $L'$-colouring of $H'$ with defect
1 extends the $L$-colouring of $G'$ to give $L$-colouring of $G$ with defect 1, which is a desired contradiction.
\end{proof}


\section{Proof of \thmref{full}}


Let $\mu:=\eg{G}$ and $t:=\bceil{2 + \sqrt{3\mu + 3}}$. If $G$ is
planar, the statement of the theorem is true by the result of Cushing and
Kierstead}\cite{DBLP:journals/ejc/CushingK10}. Thus, we may assume that
$G$ is not planar, that is $\mu>0$. Thus, $t\geq 5$. Assume for the sake
of contradiction that the statement of the theorem is false. We may
assume that $G$ is a vertex minimal, and subject to that edge maximal,
connected graph that is a counter example to the theorem. Specifically,
assume $G$ is a connected graph
with a 2-cell embedding on a surface of Euler genus at most $\mu$, and no edge can be added
to $G$ without introducing edge crossings or making $G$ non-simple, 
and $G$ is not \ch{t,1}choosable but for every $v\in V(G)$, $G- v$ is
\ch{t,1}choosable. Thus, $G$ satisfies the properties listed in \lemref{prop}.

 
Let the number of faces of the embedding of $G$ be denoted by $f$. The
Euler formula gives: $f=m-n+2-\mu$. Each face in an embedding of the
edge maximal graph has size at most $3$, 
and each edge is in 
at most 2 faces, thus $ f\leq \frac{2m}{3} $.

For each $v\in V(G)$ let its \emph{charge} $\w{v}:=\deg{v}$. 
Since $2m=\sum_{v_i\in\ V(G)}\deg{v_i}$, the Euler formula and the above inequality give 

\begin{equation}\label{eq2}
\sum_{v_i\in\ V(G)}(\w{v_i}-6)\leq 6\mu -12.
\end{equation}

We will move these charges from one vertex to another such that the overall sum $\sum_{v_i\in\ V(G)} (\w{v_i}-6)$ remains unchanged. We move the vertex charges according to the following discharging rules:\\

($\star1$) Each degree--$(t+2)$ vertex $v$ sends the charge of $\frac{1}{\floor{\frac{t}{2}}+1}$ to each degree--$t$ vertex in $N(v)$.

($\star2$) If $t$ is odd, each degree--$d$ vertex $v$, $d\geq t+3$, sends the charge of $\frac{1}{\floor{\frac{t}{2}}+1}$ to each degree--$t$ vertex in $N(v)$.

($\star3$) If $t$ is even, each degree--$d$ vertex $v$, $d\geq t+3$, sends the charge of $\frac{2}{\frac{t}{2}+1}$ to each degree--$t$ vertex in $N(v)$.\\


After above discharging rules are applied to all the vertices in $G$, their new charges are as follows. The weights of degree-$(t+1)$ vertices remain unchanged. 

Consider now the new weight of a degree-$t$ vertex $v$. By property \ref{degt} in \lemref{prop}, $v$ has at
least $\ceil{\frac{t}{2}}$ degree-$d$, $d\geq t+2$ neighbours. Thus,
if  $t$ is odd, each degree-$t$ vertex $v$ has new weight $\w{v}\geq t+
\ceil{\frac{t}{2}}\frac{1}{\floor{\frac{t}{2}}+1}=t+\frac{t+1}{2}\frac{1}{\frac{t-1}{2}+1}=t+1$. If
$t$ is even, we have two cases to consider. First consider the case
that $v$ has at last $\frac{t}{2} +1$ degree-$d$, $d\geq t+2$
neighbours. Then the new weight of $v$ is $\w{v}\geq t+
(\frac{t}{2}+1)\frac{1}{\frac{t}{2}+1}=t+1$. The second case to
consider is that $v$ has exactly $\frac{t}{2}$ degree-$d$, $d\geq t+2$
neighbours, in which case by property \ref{big} in \lemref{prop}, one
of the neighbours if $v$ has degree at least $t+3$. Then by the rule
($\star3$), the new weight of $v$ is $\w{v}\geq t+
(\frac{t}{2}-1)\frac{1}{\frac{t}{2}+1} + \frac{2}{\frac{t}{2}+1}=t+1$.

Since the minimum degree in $G$ is $t$, by Property \ref{min} of \lemref{prop}, it remains to consider the weights of the degree-$d$, $d\geq t+2$ vertices. By property \ref{degt1} in \lemref{prop}, each degree-$(t+2)$ vertex has at most $\floor{\frac{t+2}{2}}$ degree-$t$ neighbours. The rule ($\star1$) applies to each such vertex $v$ and thus its new weight is $\w{v}\geq t+2- \floor{\frac{t+2}{2}}\frac{1}{\floor{\frac{t}{2}}+1}=t+1$ (both when $t$ is odd and even).

Consider now a degree-$d$ vertex $v$, $d\geq t+3$. If $t$ is odd, then the rule ($\star2$) applies and by property \ref{degt1} in \lemref{prop}, the new weight of $v$ is  $\w{v}\geq d - \floor{\frac{d}{2}}\frac{1}{\floor{\frac{t}{2}}+1}> t+1$. (That is because, $d - \floor{\frac{d}{2}}\frac{1}{\floor{\frac{t}{2}}+1}\geq d-\frac{d}{2}\frac{1}{\frac{t-1}{2}+1}=  d-\frac{d}{t+1}$. Now $d-\frac{d}{t+1}> t+1$ whenever $d> t+2+\frac{1}{t}$. Since $t> 1$, that is our case.)

If $t$ is even, then the rule ($\star3$) applies and by property
\ref{degt1} in \lemref{prop}, the new weight of $v$ is  $\w{v}\geq d -
\floor{\frac{d}{2}}\frac{2}{\frac{t}{2}+1}$. If $d=t+3$, then $d$ is
odd and thus $\w{v}\geq t+3 -
\frac{t+2}{2}\frac{2}{\frac{t}{2}+1}=t+1$. Else $d\geq t+4$, and
$\w{v}\geq d-\floor{\frac{d}{2}}\frac{2}{\frac{t}{2}+1}> t+1$. (That
is because, $d-\floor{\frac{d}{2}}\frac{2}{\frac{t}{2}+1}\geq
d-\frac{2d}{t+2}$. Now $d-\frac{2d}{t+2}> t+1$ whenever $d>
t+3+\frac{2}{t}$. Since $t> 2$, that is our case.)

Therefore, after discharging every vertex has charge at least $t+1$.  Finally, we show that there is at least one vertex with the charge greater than $t+1$. By the above arguments each degree $d$, $d\geq t+4$, vertex $v$ has weight $\w{v}\geq d-\frac{d}{t+1}$ when $t$ is odd, and weight $\w{v}\geq d-\frac{2d}{t+2}$ when $t$ is even. Thus, in either case,  $\w{v}\geq \frac{dt}{t+2}$. By \lemref{lovaz}, $G$ has a vertex $v$ of degree at least $2t$. Since $t\geq 4$, $2t\geq t+4$. Therefore, $G$ has a vertex $v$ with weight $\w{v}\geq \frac{dt}{t+2}\geq \frac{2t^2}{t+2}$.

The new weights and the inequality \ref{eq2} give



\begin{center}
$\begin{array}{lll}
\sum_{i=1}^{n}(\w{v_i}-6) & \geq & (n-1)(t+1-6) + \frac{2t^2}{t+2}-6\\
\ & = & (n+1)(t-5) + \frac{8}{t+2}  \ \\
\ & > & (n+1)(t-5).  \ \\

\end{array} $
\end{center}

By Property \ref{numv} of \lemref{prop}, $n\geq 2t+1$, and with $t\geq
5$ and inequality \ref{eq2} we get
\begin{center}
$\begin{array}{lll}
(2t+2)(t-5) < & \sum_{i=1}^{n}(\w{v_i}-6) & \leq  6\mu -12\\
\end{array} $
\end{center}

The inequality $(2t+2)(t-5) <  6\mu -12$ is only true for $t<2 +\sqrt{3\mu + 3}$ thus for $t= \bceil{2 + \sqrt{3\mu + 3}}$ we get a
  contradiction thereby completing the proof.


\EndProof


{\small
\bibliographystyle{plain}
\bibliography{summary}
}

\end{document}